\documentclass[prl,twocolumn, aps,preprintnumbers,showpacs, superscriptaddress]{revtex4-1}

\usepackage{latexsym}
\usepackage{amsmath}
\usepackage{amssymb}
\usepackage{amsfonts}
\usepackage{ifthen}
\usepackage{revsymb}
\usepackage{yfonts}
\usepackage{graphicx}

\usepackage{graphicx}
\usepackage{amsmath}
\usepackage{amssymb}
\usepackage{amsthm}

\newcommand{\be}{\begin{equation}}
\newcommand{\ee}{\end{equation}}
\newcommand{\ba}{\begin{array}}
\newcommand{\ea}{\end{array}}
\newcommand{\bea}{\begin{eqnarray}}
\newcommand{\eea}{\end{eqnarray}}

\newcommand{\calH}{{\cal H }}

\newcommand{\calM}{{\cal M }}

\newcommand{\calT}{{\cal T }}

\newcommand{\calP}{{\cal P }}

\newcommand{\calD}{{\cal D }}

\newcommand{\CC}{\mathbb{C}}

\newcommand{\la}{\langle}
\newcommand{\ra}{\rangle}
\newcommand{\eff}{\mathrm{eff}}
\newcommand{\nn}{\nonumber}
\newcommand{\trace}{\mathop{\mathrm{Tr}}\nolimits}

\newtheorem{lemma}{Lemma}

\newtheorem{theorem}{Theorem}

\begin{document}

\title{Criticality without frustration for quantum spin-$1$ chains }

\author{Sergey \surname{Bravyi}}
%\email{sbravyi@us.ibm.com}
\affiliation{IBM Watson Research
Center,  Yorktown Heights  NY 10598}
\author{Libor \surname{Caha}}
%\email{cahalibor@me.com}
\affiliation{Faculty of Informatics, Masaryk University, Botanická 68a, Brno, Czech Republic}
\author{Ramis \surname{Movassagh}}
%\email{ramis@math.mit.edu}
\affiliation{Department of Mathematics, Massachusetts Institute of Technology, Cambridge MA 02139}
\author{Daniel \surname{Nagaj}}
%\email{daniel.nagaj@savba.sk}
\affiliation{Research Center for Quantum Information, Slovak
Academy of Sciences, Bratislava, Slovakia}
\author{Peter W. \surname{Shor}}
%\email{shor@math.mit.edu}
\affiliation{Department of Mathematics, Massachusetts Institute of Technology, Cambridge MA 02139}

\date{\today}

\begin{abstract}
Frustration-free (FF) spin chains have a property that their ground state
minimizes all individual terms in the chain Hamiltonian.
   We ask how entangled the ground state of a FF quantum spin-$s$ chain
with nearest-neighbor interactions can be for small values of $s$.
While FF spin-$1/2$ chains are known to have unentangled ground states, the
case $s=1$ remains less explored.
We propose the first example of a FF translation-invariant spin-$1$ chain
that has a unique highly entangled ground state and exhibits some signatures of
a critical behavior. The ground state can be viewed as the uniform superposition of
balanced strings of left and right parentheses
separated by empty spaces.
Entanglement entropy of one half of the chain
scales as  $\frac12 \log{n}+O(1)$, where $n$ is the number of spins. We prove that the energy gap above the ground
state is polynomial in $1/n$. The proof
relies on a new result concerning statistics of
Dyck paths  which might be of independent interest.
\end{abstract}


\maketitle

The presence of long-range entanglement in the ground states of critical spin chains
with only short-range interactions
is one of the most fascinating discoveries in the theory of quantum phase transitions~\cite{Sachdev11,Vidal03,Korepin04}.
It can be quantified by the scaling
law $S(L)\sim \log{L}$, where
$S(L)$ is the entanglement entropy of
a block of $L$ spins. In contrast, non-critical
spin chains characterized by a non-vanishing energy gap
obey an area law~\cite{Hastings07,Eisert08,Arad11}
asserting that $S(L)$ has a constant upper bound independent of $L$.

One can ask how stable is the long-range ground state entanglement
against small variations of Hamiltonian parameters?  The scaling theory predicts~\cite{Vidal03,Latorre03} that
a chain whose Hamiltonian is controlled by some parameter $g$
follows the  law $S(L)\sim \log{L}$
only if $L$ does not exceed the correlation length
$\xi\sim |g-g_c|^{-\nu}$, where $\nu>0$ is the critical exponent
and $g_c$ is the critical point. For larger $L$ the entropy $S(L)$
saturates at a constant value. Hence achieving the scaling $S(L)\sim \log{L}$
requires fine-tuning of the parameter $g$ with
precision scaling polynomially with $1/L$ posing a serious experimental challenge.

The stringent precision requirement described above can be partially avoided
for spin chains described by {\em frustration-free} Hamiltonians.
Well-known (non-critical) examples of such Hamiltonians are the
Heisenberg ferromagnetic chain~\cite{Koma95}, the AKLT model~\cite{AKLT87},
and parent Hamiltonians of matrix product states~\cite{Fannes92,PEPS07}.
More generally, we consider Hamiltonians of a form  $H=\sum_j g_j \Pi_{j,j+1}$, where
$\Pi_{j,j+1}$ is a projector acting on spins $j,j+1$ and
$g_j>0$ are some coefficients. The Hamiltonian is called frustration-free (FF)
if the projectors $\Pi_{j,j+1}$ have a common zero eigenvector $\psi$.
Such zero eigenvectors $\psi$ span the ground subspace of $H$.
Clearly, the ground subspace  does not depend on the
coefficients $g_j$ as long as they remain positive.
This inherent stability against variations of the Hamiltonian parameters
motivates a question  of whether FF Hamiltonians can describe critical spin chains.

In this Letter  we propose a  toy model describing a FF translation-invariant spin-$1$ chain
with open boundary conditions
 that has a unique ground state with a logarithmic
scaling of entanglement entropy and a polynomial energy gap.
Thus our FF model reproduces some of the main signatures of critical spin chains.
In contrast, it was recently shown by Chen et al~\cite{Chen04} that any FF spin-$1/2$ chain has an unentangled ground state.
Our work may also offer valuable insights for the problem of realizing
long-range entanglement in open quantum systems with an engineered dissipation.
Indeed, it was shown by Kraus et al~\cite{Kraus08} and Verstraete et al~\cite{Verstraete08}
that the ground state of a FF Hamiltonian can be represented as a unique steady state of a
dissipative process described by the Lindblad equation with local quantum jump operators. A proposal for realizing
such dissipative processes in cold atom systems has been made by Diehl et al~\cite{Diehl08}.

{\em Main results. \hspace{2mm}}
We begin by describing the ground state of our model.
The three basis states of a single spin will be identified
with a left bracket $l\equiv [$, right bracket $r\equiv \, ]$, and
an empty space represented by $0$. Hence a state of a single
spin can be written as $\alpha |0\ra + \beta |l\ra + \gamma |r\ra$
for some complex coefficients $\alpha,\beta,\gamma$.
For a chain of $n$ spins, basis states $|s\ra$ correspond to strings
$s\in \{0,l,r\}^n$.
A string $s$ is called a {\em Motzkin path}~\cite{Motzkin2,*Motzkin1}
iff (i) any initial  segment of $s$ contains at least as many $l$'s as $r$'s,
and (ii) the total number of $l$'s is equal to the total number of $r$'s.
For example, a string $lllr0rl0rr$ is a Motzkin path while $l0lrrrllr$ is not
since its initial segment $l0lrrr$ has more $r$'s than $l$'s.
By ignoring all $0$'s one can view Motzkin paths as  balanced strings of
left and right brackets.  We shall be interested in the {\em Motzkin state}
$|\calM_n\ra$ which is the uniform superposition of all
Motzkin paths of length $n$.
For example, $|\calM_2\ra \sim |00\ra+|lr\ra$,
$|\calM_3\ra \sim |000\ra+|lr0\ra + |l0r\ra + |0lr\ra$, and
\bea
|\calM_4\ra &\sim & |0000\ra + |00lr\ra + |0l0r\ra + |l00r\ra \nn \\
&& + |0lr0\ra + |l0r0\ra + |lr00\ra + |llrr\ra + |lrlr\ra. \nn
\eea
Let us first ask how entangled is the Motzkin state.
For a contiguous block of spins $A$,
let $\rho_A =\trace_{j\notin A} |\calM_n\ra\la \calM_n|$
be the reduced density matrix of $A$.
Two important measures of entanglement are
the Schmidt rank $\chi(A)$ equal to the number of non-zero
eigenvalues of $\rho_A$, and the
 entanglement entropy
$S(A)=-\trace \rho_A \log_2{\rho_A}$.
We will choose $A$ as the left half of the chain,
$A=\{1,\ldots,n/2\}$. We show that
\be
\label{ent}
\chi(A)=1+n/2 \quad \mbox{and} \quad S(A)= \frac12 \log_2{n} + c_n
\ee
where $\lim_{n\to \infty} c_n= 0.14(5)$.
The linear scaling of the Schmidt rank stems from the presence of locally unmatched
left brackets in $A$ whose matching right brackets belong to the complementary region $B=[1,n]\backslash A$.
The number of the locally unmatched brackets $m$ can vary from $0$ to $n/2$ and must
be the same in $A$ and $B$ leading to long-range
entanglement between the two halves of the chain.

Although the definition of Motzkin paths may seem very non-local,
we will show that the state $|\calM_n\ra$ can be specified by
imposing  local constraints on nearest-neighbor spins.
Let $\Pi$ be a projector onto the three-dimensional
subspace of $\CC^3\otimes \CC^3$ spanned by
states $|0l\ra-|l0\ra$, $|0r\ra-|r0\ra$, and $|00\ra-|lr\ra$.
Our main result is the following.
\begin{theorem}
\label{thm:main}
The Motzkin state $|\calM_n\ra$ is a unique
ground state with zero energy of a  frustration-free Hamiltonian
\be
\label{H}
H=|r\ra\la r|_1 + |l\ra\la l|_n+\sum_{j=1}^{n-1} \Pi_{j,j+1},
\ee
where  subscripts indicate spins acted upon by a projector.
The spectral gap\footnote{Here and below the spectral gap of a Hamiltonian
means the difference between the smallest and the second smallest eigenvalue.}
 of $H$ scales  polynomially with $1/n$.
 \end{theorem}
The theorem remains true if $H$
is modified by introducing arbitrary weights $g_j\ge 1$
for every projector in Eq.~(\ref{H}).
A polynomial lower bound on the spectral  gap of $H$
is, by far, the most difficult  part of Theorem~\ref{thm:main}.
Our proof consists of several steps.
First, we use a perturbation theory to relate the
spectrum of $H$ to the one of
an effective Hamiltonian $H_{\eff}$ acting on
Dyck paths --- balanced strings of left and right brackets~\footnote{One can regard Dyck paths
as a special case of Motzkin paths in which no `$0$' symbols are allowed.}.
This step involves successive applications of the Projection Lemma due to Kempe et al~\cite{KKR04}.
Secondly, we map $H_{\eff}$ to a stochastic matrix $P$ describing a random
walk on Dyck paths in which transitions correspond to insertions/removals of
consecutive $lr$ pairs.
The key step of the proof is to show that the random walk on Dyck paths is rapidly mixing.
Our method of proving the desired rapid mixing property employs the polyhedral description of matchings in bipartite
graphs~\cite{Schrijver}. This method appears to be new and might be interesting on its own right.
Exact diagonalization performed for short chains suggests that
the spectral gap of $H$ scales as $\Delta\sim 1/n^3$, see Fig.~\ref{fig:gap}.
Our proof gives an upper bound $\Delta=O(n^{-1/2})$ and a lower bound $\Delta=\Omega(n^{-c})$
for some $c\gg 1$.

\begin{center}
\begin{figure}[t]
\centerline{
\mbox{
 \includegraphics[height=5cm]{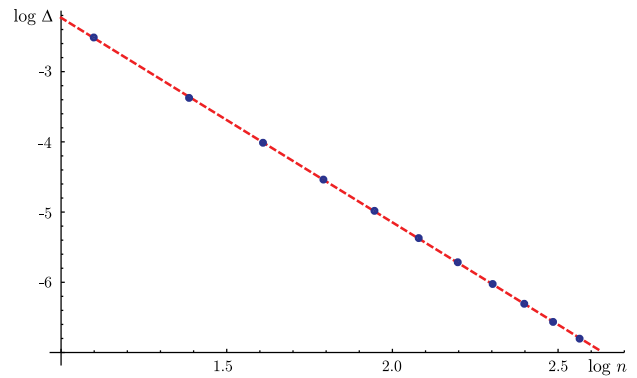}
 }}
\caption{The spectral gap $\Delta$ of the Hamiltonian $H$ defined in Eq.~(\ref{H})
for $3\le n\le 13$ obtained by the exact diagonalization. The dashed line shows a linear fit $\log{\Delta} = 0.68-2.91 \log{n}$.
Our numerics suggests that the first excited state of $H$ belongs to the subspace spanned by
strings with exactly one unmatched bracket.
}
 \label{fig:gap}
\end{figure}
\end{center}

{\em Previous work.  \hspace{2mm}}
%SBB3: clarification regarding Irani09 and Gottesman09 + some reordering.
%Also I realized that in the previous version d was used without definition.
Examples of spin chain Hamiltonians with highly entangled ground states
have been constructed by Gottesman and Hastings~\cite{Gottesman09},
and Irani~\cite{Irani09} for local dimension $d=9$ and $d=21$ respectively
(here and below $d\equiv 2s+1$).
 These models exhibit a linear scaling of the entropy $S(L)$ for some blocks of spins
while the spectral gap is polynomial in $1/n$. The model found in~\cite{Irani09}
is FF and translation-invariant.
Ref.~\cite{MFGNOS} focused on `generic' spin chains with a Hamiltonian $H=\sum_j \Pi_{j,j+1}$
where the projectors  $\Pi_{j,j+1}$ are chosen randomly with a fixed rank $r$~\footnote{Though the results
of Ref.~\cite{MFGNOS} are applicable to more general Hamiltonians, the
convenient restriction to random projectors is sufficient for
addressing the degeneracy and frustration condition.}.
The authors of~\cite{MFGNOS} identified three important regimes: (i) frustrated chains, $r >d^2/4$,
(ii) FF chains, $d\le r\le d^2/4$, and (iii) FF chains with product ground states, $r<d$.
%SBB3: added `generic' and ``with probability one'', as Ramis suggested.
It was conjectured in~\cite{MFGNOS} that generic FF chains in the regime $d\le r\le d^2/4$
have only highly entangled ground states with probability one.
This regime however requires local dimension $d\ge 4$.
The new model based on Motzkin paths corresponds to the case $d=r=3$ (ignoring the boundary terms)
and thus it can be frustrated by arbitrarily small deformations of the projectors making them generic.
In addition, results of~\cite{MFGNOS} imply that
examples of FF spin-$1$ chains with highly entangled ground states have measure zero
in the parameter space.
The question of whether matrix product states specified by FF parent Hamiltonians can exhibit
quantum phase transitions has been studied by Wolf et al~\cite{Wolf06}.
However, the models studied
in~\cite{Wolf06} have bounded entanglement entropy, $S(L)=O(1)$.
%SBB3:  I changed ``do not have highly entangled ground states"
%because it sounds a bit negative.

{\em Hamiltonian.  \hspace{2mm}}
Let us now construct a FF Hamiltonian $H$ whose unique ground state is $|\calM_n\ra$.
First we need to find a more local description of Motzkin paths.
Let $\Sigma=\{0,l,r\}$. We will say that a pair of strings $s,t\in \Sigma^n$
is equivalent, $s\sim t$, if $s$ can be obtained from $t$ by a sequence of local
moves
\be
\label{moves}
00\leftrightarrow lr, \quad 0l\leftrightarrow l0, \quad 0r\leftrightarrow r0.
\ee
These moves can be applied to any consecutive pair of letters.
For any integers $p,q\ge 0$ such that $p+q\le n$ define a string
\[
c_{p,q}\equiv \underbrace{r\ldots r}_p \underbrace{0\ldots 0}_{n-p-q} \underbrace{l\ldots l}_q.
\]
\begin{lemma}
Any string $s\in \Sigma^n$ is equivalent to one and only one string $c_{p,q}$.
A string $s\in \Sigma^n$ is a Motzkin path iff it is equivalent to the all-zeros string,
$s\sim c_{0,0}$.
\end{lemma}
\begin{proof}
Indeed, applying the local moves Eq.~(\ref{moves}) one can make sure that $s$
does not contain substrings $lr$ or $l0\ldots 0r$. If this is the case and $s$ contains at least one $l$,
then all letters to the right of $l$ are $l$ or $0$. Similarly, if $s$ contains at least one $r$, then all
letters to the left of $r$ are $r$ or $0$. Since we can swap $0$ with any other letter by the local moves,
$s$ is equivalent to $c_{p,q}$ for some $p,q$.
It remains to show that different strings $c_{p,q}$ are not equivalent to each other.
Let $L_j(s)$ and $R_j(s)$ be the number of $l$'s and $r$'s among the first $j$ letters of $s$.
Suppose $c_{p,q}\sim c_{p',q'}$ and $p\ge p'$.
Then $R_p(s)-L_p(s)\le p'$ for any string $s$ equivalent to $c_{p',q'}$.
This is a contradiction unless $p=p'$.  Similarly one shows that $q=q'$.
\end{proof}
The lemma shows that the set of all strings $\Sigma^n$ can be partitioned
into equivalence classes $C_{p,q}$, such that $C_{p,q}$ includes
all strings equivalent to $c_{p,q}$.
In other words, $s\in C_{p,q}$ iff $s$ has $p$ unmacthed right brackets
and $q$ unmatched left brackets. Accordingly, the set of Motzkin paths
$\calM_n$ coincides with the equivalence class $C_{0,0}$.

Let us now define projectors `implementing' the local moves in Eq.~(\ref{moves}).
Define normalized states
\[
|\phi\ra\sim |00\ra-|lr\ra, \quad |\psi^l\ra\sim |0l\ra -|l0\ra, \quad |\psi^r\ra\sim |0r\ra-|r0\ra
\]
and a projector $\Pi=|\phi\ra\la \phi|+|\psi^l\ra \la \psi^l| + |\psi^r\ra \la \psi^r|$.
Application of $\Pi$ to  a pair of spins $j,j+1$ will be denoted $\Pi_{j,j+1}$.
If some state $\psi$ is annihilated by every projector $\Pi_{j,j+1}$,
it must have the same amplitude on any pair of equivalent strings, that is,
$\la s|\psi\ra=\la t|\psi\ra$ whenever $s\sim t$. It follows that a Hamiltonian
$H_\sim=\sum_{j=1}^{n-1} \Pi_{j,j+1}$ is FF and the ground subspace of $H_\sim$
is spanned by pairwise orthogonal states
$|C_{p,q}\ra$, where $|C_{p,q}\ra$ is the uniform superposition of all
strings in $C_{p,q}$.
The desired Motzkin state $|\calM_n\ra=|C_{0,0}\ra$ is thus a ground state of $H_\sim$.
(It is worth mentioning that not all states $|C_{p,q}\ra$ are highly entangled.
For example, $|C_{n,0}\ra=|r\ra^{\otimes n}$ is a product state.)
How can we exclude the unwanted ground states $|C_{p,q}\ra$ with $p\ne 0$ or $q\ne 0$?
We note that $C_{0,0}$ is the only
class in which strings never start from $r$ and never end with $l$.
 Hence a modified Hamiltonian
$H=|r\ra\la r|_1 + |l\ra\la l|_n + H_{\sim}$
that penalizes strings starting from $r$ or ending with $l$ has
a unique ground state $|C_{0,0}\ra$. This proves the first part of Theorem~\ref{thm:main}.
We can also consider weighted Hamiltonians $H_\sim(g) =\sum_{j=1}^{n-1} g_j \Pi_{j,j+1}$
and $H(g)=g_0 |r\ra\la r|_1 + g_n |l\ra\la l|_n + H_{\sim}(g)$, where $g_0,\ldots,g_n\ge 1$
are arbitrary coefficients. One can easily check that the ground state of $H(g)$ does not
depend on $g$ and $H(g)\ge H$. It implies that the spectral gap of $H(g)$ is lower bounded
by the one of $H$.

{\em Entanglement entropy.  \hspace{2mm}}
We can now construct the Schmidt decomposition of the Motzkin state.
Let $A=\{1,\ldots,n/2\}$ and $B=\{n/2+1,\ldots,n\}$ be the two halves of the chain
(we assume that $n$ is even). For any string $s\in \Sigma^n$ let $s_A$ and $s_B$
be the restrictions of $s$ onto $A$ and $B$. We claim that $s$ is a Motzkin path iff
$s_A\sim c_{0,m}$ and $s_B\sim c_{m,0}$ for some $0\le m \le n/2$.
Indeed, $s_A$ ($s_B$) cannot have
unmatched right (left) brackets, while each unmatched left bracket in $s_A$ must be matched
with some unmatched right bracket in $s_B$. It follows that the Schmidt decomposition
of $|\calM_n\ra$ can be written as
\be
\label{Schmidt}
|\calM_n\ra = \sum_{m=0}^{n/2} \sqrt{p_m} \, |\hat{C}_{0,m}\ra_A  \otimes |\hat{C}_{m,0}\ra_B,
\ee
where $|\hat{C}_{p,q}\ra$ is the normalized uniform superposition
of all strings in $C_{p,q}$ and $p_m$
are the Schmidt coefficients  defined by
\be
\label{Schmidt1}
p_m=\frac{|C_{0,m}(n/2)|^2}{|C_{0,0}(n)|}.
\ee
Here we added an explicit dependence of the classes $C_{p,q}$ on $n$.
For large $n$ and $m$ one can use an
approximation  $p_m\sim m^2 \exp{(-3m^2/n)}$, see the Supplementary Material
for the proof. Note that $p_m$ achieves its maximum at $m^*\approx \sqrt{n/3}$.
Approximating the sum $\sum_m p_m=1$ by an integral over $\alpha=m/\sqrt{n}$
one gets  $p_m\approx n^{-1/2} q_{\alpha(m)}$, where $q_\alpha$ is a normalized pdf defined as
\[
q_\alpha=Z^{-1} \alpha^2 e^{-3\alpha^2}, \quad
Z\equiv \int_0^\infty d\alpha \alpha^2 e^{-3\alpha^2}=\frac{\sqrt{\pi}}{4\cdot 3^{3/2}}.
\]
It gives
\[
S(A)=-\sum_m p_m \log_2{p_m} \approx \log{\sqrt{n}} -\int_0^\infty d\alpha \, q_\alpha \log_2{q_\alpha}.
\]
Evaluating the integral over $\alpha$ yields Eq.~(\ref{ent}).
The approximation $p_m\approx n^{-1/2} q_{\alpha(m)}$ also implies
that $\max_m p_m=O(n^{-1/2})$. This bound will be used below in our spectral gap analysis.
We conjecture  that one can achieve a power law
scaling of $S(A)$ in Eq.~(\ref{ent}) by introducing two types of brackets, say $l\equiv [$, $r\equiv ]$,
$l'\equiv \{$, and $r'\equiv \}$, such that bracket pairs $lr$ and $l'r'$ are created
from the `vacuum' $00$ in a maximally entangled state $(|lr\ra+|l'r'\ra)/\sqrt{2}$.
The local moves Eq.~(\ref{moves}) must be modified as
$0x \leftrightarrow x0$, where $x$ can be either of $l,r,l',r'$, and
$00\leftrightarrow (lr+l'r')/\sqrt{2}$.
We expect the modified model with two types of brackets to obey a scaling
$S(A)\sim \sqrt{n}$, while its gap will remain lower bounded by an inverse  polynomial.

{\em Spectral gap: upper bound. \hspace{2mm}} Let $\lambda_2>0$ be the smallest non-zero eigenvalue of the Hamiltonian defined in Eq.~(\ref{H}).
We shall use the fact that the ground state $|\calM_n\ra$ is highly entangled to
prove an upper bound  $\lambda_2\le O(n^{-1/2})$.
Fix any $k\in [0,n/2]$ and define a `twisted' version of the ground state:
\[
|\phi\ra =\sum_{m=0}^{n/2} e^{i\theta_m} \, \sqrt{p_m} \, |\hat{C}_{0,m}\ra_A  \otimes |\hat{C}_{m,0}\ra_B,
\]
where $\theta_m=0$ for $0\le m\le k$ and $\theta_m=\pi$ otherwise.
Note that $|\phi\ra$ and $|\calM_n\ra$ have the same reduced density matrices on $A$ and $B$.
Hence $\la \phi|H|\phi\ra =\la \phi|\Pi_{cut}|\phi\ra$, where $\Pi_{cut}\equiv \Pi_{n/2,n/2+1}$.
Since $\max_m p_m =O(n^{-1/2})$ and $\sum_m p_m=1$,  there must exist $k\in [0,n/2]$
such that $\sum_{0\le m\le k} p_m =1/2 \pm \epsilon$ for some $\epsilon=O(n^{-1/2})$.
This choice of $k$ ensures that
 $\la \phi|\calM_n\ra = \sum_m p_m e^{i\theta_m} \le 2\epsilon$,
 that is $\phi$ is almost orthogonal to the ground state.
Define a normalized state $|\tilde{\phi}\ra\sim |\phi\ra- \la \calM_n|\phi\ra \cdot |\calM_n\ra$.
Then $\la \tilde{\phi}|\calM_n\ra=0$ and
$\la \tilde{\phi}|H|\tilde{\phi}\ra =\la \tilde{\phi}|\Pi_{cut}|\tilde{\phi}\ra \le \la \phi|\Pi_{cut}|\phi\ra+O(\epsilon)$.
The difference $\la \phi|\Pi_{cut}|\phi\ra-\la \calM_n|\Pi_{cut}|\calM_n\ra$ gets contributions
only from the terms $m=k,k\pm 1$  in the Schmidt decomposition,  since $\Pi_{cut}$ can change the number of unmatched brackets in $A$
and $B$ at most by one. Since  $\la \calM_n|\Pi_{cut}|\calM_n\ra=0$, we get
\[
\la \phi|\Pi_{cut}|\phi\ra \le O(p_k+p_{k-1}+p_{k+1})=O(n^{-1/2}).
\]
We arrive at $\la \tilde{\phi}|H|\tilde{\phi}\ra=O(n^{-1/2})$.
Therefore $\lambda_2$ is at most $O(n^{-1/2})$.

{\em Spectral gap: lower bound.  \hspace{2mm}}
It remains to prove a lower bound $\lambda_2\ge n^{-O(1)}$.
Let $\calH_{p,q}$ be the subspace spanned by strings $s\in C_{p,q}$
and $\calH_M\equiv \calH_{0,0}$ be the Motzkin space spanned by Motzkin paths.
Note that $H$ preserves any subspace $\calH_{p,q}$ and the unique
ground state of $H$ belongs to $\calH_M$. Therefore it suffices
to derive a lower bound $n^{-O(1)}$ for two quantities: (i) the gap of $H$ inside the Motzkin space $\calH_M$,
and (ii) the ground state energy of $H$ inside any `unbalanced' subspace $\calH_{p,q}$
with $p\ne 0$ or $q\ne 0$. Below we shall focus on part~(i)
since it allows us to introduce all essential ideas. The proof of part~(ii) can be found in the Supplementary Material.

Recall that a string $s\in \{l,r\}^{2m}$
is called a {\em Dyck path} iff any initial  segment of $s$ contains at least as many $l$'s as $r$'s,
and the total number of $l$'s is equal to the total number of $r$'s.
For example, Dyck paths of length $6$ are $lllrrr$, $llrlrr$, $llrrlr$, $lrlrlr$, and $lrllrr$.
The proof of part~(i) consists of the following steps:\\
{\em Step~1.} Map the original Hamiltonian $H$ acting on Motzkin paths
to an effective Hamiltonian $H_{\eff}$ acting on Dyck paths using perturbation
theory. \\
{\em Step~2.} Map $H_{\eff}$ to a stochastic matrix $P$ describing a random walk on Dyck paths
in which transitions correspond to insertions or removals of consecutive $lr$ pairs.\\
{\em Step~3.} Bound the spectral gap of $P$
using the canonical paths method~\cite{Diaconis91,*Sinclair92}.

To construct a good family of canonical paths in Step~3 we will
organize Dyck paths into a rooted tree in which
level-$m$ nodes represent Dyck paths of length $2m$,
edges correspond to insertion of $lr$ pairs, and each node has at most four children.
Existence of such a tree will be proved using  the fractional matching method~\cite{Schrijver}.

Let $\calD_m$ be the set of Dyck paths of length $2m$,
$\calD$ be the union of all $\calD_m$ with $2m\le n$,
and $\calM_n$ be the set of Motzkin paths of length $n$.
Define a Dyck space $\calH_D$ whose basis vectors are Dyck paths $s\in \calD$.
Given a Motzkin path $u$ with $2m$ brackets, let $\mathrm{Dyck}(u)\in \calD_{m}$
be the Dyck path obtained from $u$ by removing zeros. We shall use an embedding
$V\, : \, \calH_D \to \calH_{M}$ defined as
\[
V\, |s\ra = \frac1{\sqrt{{n \choose 2m}}}\, \sum_{\substack{u\in \calM_n\\ \mathrm{Dyck}(u)=s\\}}\, |u\ra, \quad \quad s\in
\calD \cap \calD_m.
\]
One can easily check that $V^\dag V=I$, that is, $V$ is an isometry.
For any Hamiltonian $H$, let $\lambda_2(H)$ be the second smallest
eigenvalue of $H$.

\noindent
{\em Step~1.} The restriction of the Hamiltonian Eq.~(\ref{H}) onto the Motzkin space $\calH_M$
can be written as $H=H_{move}+H_{int}$,
where $H_{move}$ describes freely moving left and right brackets, while
$H_{int}$ is an `interaction term' responsible for pairs creation.
More formally, $H_{move}=\sum_{j=1}^{n-1} \Pi_{j,j+1}^{move}$ and
$H_{int}=\sum_{j=1}^{n-1} \Pi_{j,j+1}^{int}$, where
$\Pi^{move}$ projects onto the subspace spanned by $|0l\ra-|l0\ra$
and $|0r\ra - |r0\ra$, while $\Pi^{int}$
projects onto the state $|00\ra-|lr\ra$.
Note that the boundary terms in Eq.~(\ref{H}) vanish on $\calH_M$.
 We shall treat $H_{int}$ as a small
perturbation of $H_{move}$.  To this end define a modified FF Hamiltonian
$H_{\epsilon}=H_{move} + \epsilon H_{int}$,
where $0<\epsilon\le 1$ will be chosen later.
One can easily check that $|\calM_n\ra$ is the unique ground state of $H_\epsilon$
and $\lambda_2(H)\ge \lambda_2(H_\epsilon)$ (use the operator inequality $H\ge H_\epsilon$).
Note that $H_{move}\psi=0$ iff $\psi$ is symmetric under the moves $0l\leftrightarrow l0$
and $0r\leftrightarrow r0$. It follows that the ground subspace of $H_{move}$
is spanned by states $V\, |s\ra$ with $s\in \calD$.
To compute the spectrum of $H_{move}$, we can ignore the difference between
$l$'s and $r$'s since $H_{move}$ is only capable of swapping zeros with
non-zero letters. It follows that the spectrum of $H_{move}$
must coincide with the spectrum of the Heisenberg ferromagnetic spin-$1/2$ chain,
that is, $\Pi^{move}$ can be replaced by the projector onto the singlet state $|01\ra-|10\ra$,
where $1$ represents either $l$ or $r$. Using the exact formula for
the spectral gap of the Heisenberg chain found by Koma and Nachtergaele~\cite{Koma95} we arrive at
$\lambda_2(H_{move})=1-\cos{\left(\frac{\pi}{n}\right)}=\Omega(n^{-2})$.
Let
\[
H_{\eff}=V^\dag H_{int} V
\]
be the first-order effective Hamiltonian acting on the Dyck space $\calH_D$.
Applying the Projection Lemma of~\cite{KKR04} to the orthogonal complement of $|\calM_n\ra$ in $\calH_M$ we infer that
%SBB4: a typo is fixed in Eq. (6). The second term in the rhs must contain \|H_{int}\|^2 in the numerator.
%Eq. (7). is updated accordingly
\be
\label{KKR}
\lambda_2(H_\epsilon)\ge \epsilon \lambda_2(H_{\eff})-\frac{O(\epsilon^2) \|H_{int}\|^2}{\lambda_2(H_{move})-2\epsilon\|H_{int}\| }.
\ee
Choosing $\epsilon\ll n^{-3}$ guarantees that $\epsilon\|H_{int}\|$ is small compared with $\lambda_2(H_{move})$.
For this choice of $\epsilon$ one gets
\be
\label{KKR1}
\lambda_2(H)\ge \lambda_2(H_\epsilon)\ge \epsilon \lambda_2(H_{\eff})- O(\epsilon^2 n^4).
\ee
Hence it suffices to prove that $\lambda_2(H_{\eff})\ge n^{-O(1)}$.

\noindent
{\em Step~2.} Recall that $H_{\eff}=V^\dag H_{int} V$ acts on the Dyck space $\calH_D$. Its unique
ground state $|\calD\ra\in \calH_D$ can be found by solving $|\calM_n\ra=V\,|\calD\ra$. It yields
\be
\label{Dstate}
|\calD\ra=\frac1{\sqrt{|\calM_n|}} \sum_{2m\le n} \sqrt{{n \choose 2m}} \sum_{s\in \calD_m} |s\ra.
\ee
Let $\pi(s)=\la s|\calD\ra^2$ be the induced probability distribution on $\calD$.
Given a pair of Dyck paths $s,t\in \calD$, define
\be
\label{walk}
P(s,t)=\delta_{s,t}-\frac1{n}\la s |H_{\eff} |t\ra \sqrt{\frac{\pi(t)}{\pi(s)}}.
\ee
We claim that $P$ describes a random walk on the set of Dyck paths $\calD$
such that $P(s,t)$ is a transition probability from $s$ to $t$,
and $\pi(s)$ is the unique steady state of $P$.
Indeed, since $\sqrt{\pi(t)}$ is a zero eigenvector of $H_{\eff}$, one has $\sum_t P(s,t)=1$
and $\sum_s \pi(s) P(s,t)=\pi(t)$.
Off-diagonal matrix elements $\la s|H_{\eff}|t\ra$ get contributions only from
terms $-|00\ra\la lr|$ and $-|lr\ra\la 00|$ in $H_{int}$. It implies that $\la s|H_{\eff}|t\ra\le 0$
for $s\ne t$ and hence $P(s,t)\ge 0$. Furthermore, $P(s,s)\ge 1/2$ since $\la s|H_{\eff}|s\ra \le n/2$.
In the Supplementary Material we shall prove the following.
\begin{lemma}
\label{lemma:transitions}
Let $s,t\in \calD$ be any Dyck paths such that
$t$ can be obtained from $s$  by adding or removing a single $lr$ pair. Then
$P(s,t)=\Omega(1/n^3)$. Otherwise $P(s,t)=0$.
\end{lemma}
Let $\lambda_2(P)$ be the second largest eigenvalue of $P$. From Eq.~(\ref{walk})
one gets $\lambda_2(H_{\eff})=n(1-\lambda_2(P))$.
Hence it suffices to prove that the random walk
$P$ has a polynomial spectral gap, that is, $1-\lambda_2(P)\ge n^{-O(1)}$.

\noindent
{\em Step~3.} Lemma~\ref{lemma:transitions} tells us that $P$ describes
a random walk on a graph $G=(\calD,E)$ where two Dyck
paths are connected by an edge, $(s,t)\in E$,
iff  $s$ and $t$ are related by insertion/removal of a single $lr$ pair.
To bound the spectral gap of $P$ we shall connect any pair
of Dyck paths $s,t\in \calD$ by a {\em canonical path} $\gamma(s,t)$,
that is, a sequence $s_0,s_1,\ldots,s_l\in \calD$ such that
$s_0=s$, $s_l=t$, and $(s_i,s_{i+1})\in E$ for all $i$.
The canonical paths theorem~\cite{Diaconis91,Sinclair92} shows that
$1-\lambda_2(P)\ge 1/(\rho l)$, where $l$ is the maximum length
of a canonical path and $\rho$ is the maximum edge load defined as
\be
\label{edge_load}
\rho=\max_{(a,b)\in E} \; \frac1{\pi(a) P(a,b)} \sum_{s,t\, : \, (a,b)\in \gamma(s,t)} \; \; \pi(s) \pi(t).
\ee
The key new result that allows us to choose a good family of canonical paths is the following.
\begin{lemma}
\label{lemma:map}
Let $\calD_k$ be the set of Dyck paths of length $2k$.
For any $k\ge 1$ there exists a map $f\, : \, \calD_k \to \calD_{k-1}$ such that
(i)  the image of any path $s\in \calD_k$ can be obtained from $s$ by
removing a single $lr$ pair,
(ii) any path $t\in \calD_{k-1}$ has at least one  pre-image in $\calD_k$, and
(iii) any path $t\in \calD_{k-1}$ has at most four  pre-images in $\calD_k$.
\end{lemma}
The lemma allows one to organize the set of all Dyck paths $\calD$ into a
{\em supertree}  $\calT$ such that the root of $\calT$ represents the empty path and such that children
of any node $s$ are elements of  $f^{-1}(s)$. The properties of $f$ imply that Dyck paths of length $2m$
coincide with level-$m$ nodes of $\calT$,
any step away from the root on $\calT$
corresponds to insertion of a single $lr$ pair,  and any node of $\calT$ has at most four children.
Hence the lemma provides a recipe for growing long Dyck paths
from short ones without overusing any intermediate Dyck paths. It should be noted
that restricting the maximum number of children to four is optimal
since $|\calD_k|=C_k\approx 4^k/\sqrt{\pi} k^{3/2}$, where $C_k$ is the $k$-th Catalan number.
Our proof of Lemma~\ref{lemma:map} based on the fractional matching method can be found in
the Supplementary Material.
Five lowest levels of the supertree $\calT$ are shown on Fig.~\ref{fig:supertree}.

\begin{center}
\begin{figure*}[t]
\centerline{
\mbox{
 \includegraphics[height=4.5cm]{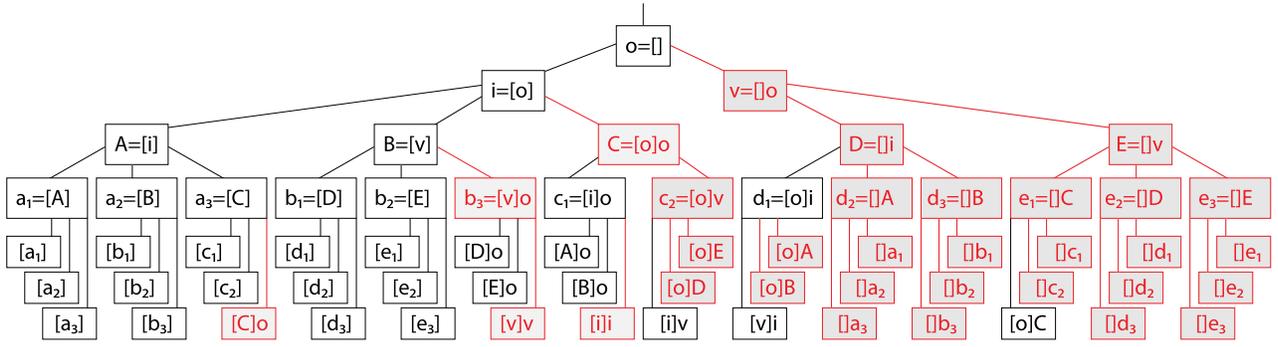}
 }}
\caption{(Color Online) Five lowest levels of the supertree $\calT$.
Nodes of $\calT$ are Dyck paths --- balanced strings of left and right brackets.
Depth-$k$ nodes are in one-to-one correspondence with Dyck paths
of length $2k$ (the set $\calD_k$). Any step towards the root requires removal of a consecutive $[\, ]$ pair.
Any node has at most four children.
The supertree can be described a family of maps $f\, :\, \calD_k\to \calD_{k-1}$
such that $f(s)$ is the parent of $s$. The maps $f$ are defined inductively such that
$f([X])=[f(X)]$, $f([\, ]Y)=[\, ]f(Y)$ for any node,  $f([X]Y)=[f(X)]Y$ for black nodes,
and $f([X]Y)=[X]f(Y)$ for red (shaded) nodes.
See the proof of Lemma~\ref{lemma:map} in the Supplementary Material for more details.
}
 \label{fig:supertree}
\end{figure*}
\end{center}

We can now define the canonical path $\gamma(s,t)$ from
$s\in \calD_m$ to $t\in \calD_k$.
Any intermediate state in $\gamma(s,t)$ will be represented as $uv$ where $u\in \calD_{l'}$ is
%SBB3: removed ``partial implementation"
an ancestor of $s$ in the supertree  and $v\in \calD_{l''}$ is an ancestor of $t$.
The canonical path starts from $u=s$, $v=\emptyset$ and
alternates between shrinking $u$ and growing $v$ by making steps towards the root  (shrink)
and away from the root (grow) on the supertree. The path terminates as soon as $u=\emptyset$ and $v=t$.
The shrinking steps are skipped whenever $u=\emptyset$, while the growing steps
are skipped whenever $v=t$.
Note that any intermediate state $uv$ obeys
\be
\label{constraint}
\min{(|s|,|t|)} \le |u|+|v| \le \max{(|s|,|t|)}.
\ee
Since any path $\gamma(s,t)$ has length at most $2n$, it suffices to bound
the maximum edge load $\rho$.
Fix the edge $(a,b)\in E$ with the maximum load.
Let $\rho(m,k,l',l'')$ be the contribution to $\rho$ that comes from canonical paths
$\gamma(s,t)$ such that $a=uv\in \calD_{l'+l''}$, where
\[
s\in \calD_{m}, \quad t\in \calD_{k}, \quad u\in \calD_{l'}, \quad v\in \calD_{l''},
\]
and such that $b$ is obtained from $a$ by growing $v$ (the case when $b$ is obtained from
$a$ by shrinking  $u$ is analogous).
The number of possible source strings $s\in \calD_m$ contributing to $\rho(m,k,l',l'')$
is at most $4^{m-l'}$ since $s$ must be a descendant of $u$ on the supertree.
The number of possible target strings $t\in \calD_k$ contributing to $\rho(m,k,l',l'')$
is at most $4^{k-l''}$ since $t$ must be a descendant of $v$ on the supertree.
Taking into account that $\pi(s)$ and $\pi(t)$ are the same for all $s\in \calD_m$
and $t\in \calD_k$ we arrive at
\[
\rho(m,k,l',l'') \le 4^{m+k-l'-l''} \frac{\pi(s)\pi(t)}{\pi(a) P(a,b)}=\frac{\pi_m \pi_k}{\pi_{l'+l''} P(a,b)}
\]
with $\pi_l=4^l {n \choose 2l}/|\calM_n|$.
 Here we used the identity
$\pi(w)=\la w|\calD\ra^2$ and Eq.~(\ref{Dstate}).
Lemma~\ref{lemma:transitions} implies that $1/P(a,b)\le n^{O(1)}$.
Furthermore, the fraction of Motzkin paths of length $n$
that have exactly $2l$ brackets is
$\sigma_l=C_l {n \choose 2l}/|\calM_n|$.
However $C_l\approx 4^l/\sqrt{\pi} l^{3/2}$
coincides with $4^l$ modulo factors polynomial in $1/n$. Hence
\[
\rho(m,k,l',l'') \le n^{O(1)} \cdot \frac{\sigma_m \sigma_k}{\sigma_{l'+l''}}.
\]
By definition, $\sigma_l\le 1$ for all $l$.
Also, one can easily check that $\sigma_l$ as a function of $l$ has a unique
maximum at $l\approx n/3$ and decays monotonically away from the maximum.
Consider two cases. {\em Case~(1):}  $l'+l''$ is on the left from the maximum of $\sigma_l$.
From Eq.~(\ref{constraint}) one gets
$\min{(m,k)}\le l'+l''$ and thus
$\sigma_m\sigma_k\le \sigma_{\min{(m,k)}} \le \sigma_{l'+l''}$.
{\em Case~(2):} $l'+l''$ is on the right from the maximum of $\sigma_l$.
From Eq.~(\ref{constraint}) one gets
$\max{(m,k)}\ge l'+l''$ and thus
$\sigma_m\sigma_k\le \sigma_{\max{(m,k)}} \le \sigma_{l'+l''}$.
In both cases we get a bound $\rho(m,k,l',l'')\le n^{O(1)}$.
Since the number of choices for $m,k,l',l''$ is at most $n^{O(1)}$,
we conclude  that $\rho\le n^{O(1)}$ and thus $1-\lambda_2(P)\ge n^{-O(1)}$.

{\em Open problems.  \hspace{2mm}}
Our work raises several questions. First, one can ask
what is the upper bound on the ground state entanglement
of FF spin-$1$ chains and whether the Motzkin state achieves this bound.
For example, if the Schmidt rank $\chi(L)$ for a block of $L$ spins can only grow
polynomially with $L$, as it is the case for the Motzkin state, ground states
of FF spin-$1$ chains could be efficiently represented by Matrix Product States~\cite{Verstraete05}
%SBB3: a new remark
(although finding such representation might be a computationally hard problem~\cite{Schuch08}).
One drawback of the model based on Motzkin paths is the need for boundary conditions
and the lack of the thermodynamic limit. It would be interesting to find examples of FF spin-$1$
chains with highly entangled ground states that are free from this drawback.
We also leave open the question of whether our model can indeed be regarded as a critical spin chain
in the sense that its continuous limit can be described by a conformal field theory.
Finally, an intriguing open question is whether long-range ground state entanglement (or steady-state
entanglement in the case of dissipative processes)
in 1D spin chains can be stable against more general local perturbations, such as external magnetic fields.

\section{Supplementary Material}

\subsection{Schmidt coefficients of the Motzkin state}

In this section we compute the Schmidt coefficients $p_m$ defined in Eq.~(\ref{Schmidt1})
and show that for large $n$ and $m$ one can use an approximation
\be
\label{approx}
p_m\sim m^2 \exp{(-3m^2/n)}.
\ee
Let $\calD_{n,k}\subseteq \{l,r\}^{2n+k}$ be the set of balanced
strings of left and right brackets of length $2n+k$ with $k$ extra left
brackets. More formally, $s\in \calD_{n,k}$ iff
any initial  segment of $s$ contains at least as many $l$'s as $r$'s,
and the total number of $l$'s is equal to $n+k$.
\begin{lemma}[\bf Andr\'e's reflection method]
\label{lemma:reflection}
The total number of strings in $\calD_{n,k}$ is
\[
D_{n,k}=\frac{k+1}{n+k+1} {2n+k \choose n}.
\]
\end{lemma}
\begin{proof}
For any bracket string $s$ let $L(s)$ and $R(s)$ be the number of left and right brackets in $s$.
Any $s\in  \{l,r\}^{2n+k}$ such that $s\notin \calD_{n,k}$ can be uniquely
represented as $s=u r v$, where $r$ corresponds to the first unmatched right bracket
in $s$, while $u$ is a balanced string (Dyck path). Let $v'$ be a string obtained from $v$ by a reflection $r\leftrightarrow l$
and $s'=urv'$. Then
\[
L(s')=L(u)+L(v')=R(u)+R(v)=R(s)-1=n-1
\]
and
\bea
R(s')&=& R(u)+1+R(v')=L(u)+1+L(v)=L(s)+1 \nn \\
&=& n+k+1. \nn
\eea
Furthermore, any string $s'$ with $n-1$ left brackets and $n+k+1$ right bracket
can be uniquely represented as $s'=urv'$. Hence the number of strings in $\calD_{n,k}$ is
\[
D_{n,k}={2n+k \choose n} -{2n+k \choose n-1}=\frac{k+1}{n+k+1} {2n+k \choose n}.
\]
\end{proof}
One can easily check that $|C_{p,q}(n)|=|C_{0,p+q}(n)|$ since the identity of unmatched brackets
does not matter for the counting.
Let
\[
M_{n,m}\equiv |C_{0,m}(n)|.
\]
Lemma~\ref{lemma:reflection} implies that
\[
M_{n,m}=
\sum_{\substack{i\ge 0 \\ 2i+m\le n\\}} \;  \frac{m+1}{i+m+1} {n \choose 2i+m} {2i+m \choose i}.
\]
It can be rewritten as
\be
\label{Mnm}
M_{n,m}=\sum_{\substack{i\ge 0 \\ 2i+m\le n\\}} M_{n,m,i},
\ee
where
\be
M_{n,m,i}=\frac{(m+1)\cdot n!}{(i+m+1)!i!(n-2i-m)!}.
\ee
Let $\alpha=m/\sqrt{n}$ and $\beta=(i-n/3)/\sqrt{n}$.
Using Stirling's formula one can get
\be
M_{n,m,i}\approx
\frac{3\sqrt{3}}{2\pi n^{3/2}}3^{n+1}\alpha\exp\left(-3\alpha^{2}-9\alpha\beta-9\beta^{2}\right).
\ee
We approximate the sum in Eq.~(\ref{Mnm}) by integrating over $i$. Since  $i=\frac{n}{3}+\beta\sqrt{n}$,
we get $di=\sqrt{n}d\beta$, Since the maximum is near $i=\frac{n}{3}$,
we can turn this sum into an integral from $-\infty$ to $\infty$.
The integral we need to evaluate is thus
\begin{eqnarray*}
M_{n,m} & \approx & \frac{3\sqrt{3}}{2\pi n^{3/2}}3^{n+1}\alpha\int_{-\infty}^{\infty}e^{-3\alpha^{2}-9\alpha\beta-9\beta^{2}}\sqrt{n}d\beta\\
 & = & \frac{3\sqrt{3}}{2\pi n}3^{n+1}\alpha\int_{-\infty}^{\infty} e^{-9\left(\beta-\alpha/2\right)^{2}-\frac{3}{4}\alpha^{2}}\sqrt{n}d\beta\\
 & = & \frac{\sqrt{3}}{2\sqrt{\pi}n}3^{n+1}\alpha e^{-\frac{3}{4}\alpha^{2}}
 \sim m \exp{(-3m^2/4n)}.
\end{eqnarray*}
Recalling that
\be
p_m = \frac{|C_{0,m}(n/2)|^2}{|C_{0,0}(n)|} \sim M_{n/2,m}^2
\ee
 we arrive at Eq.~(\ref{approx}).

\subsection{Proof of Lemma~\ref{lemma:transitions}}
Suppose $s\in \calD_k$ and $t\in \calD_{k\pm1}$.
Using the definition of $P(s,t)$ one can easily get
\[
P(s,t)=-\frac1{n}{n \choose 2k}^{-1} \sum_{u\in \calM_n[s]}\; \;  \sum_{v\in \calM_n[t]} \la u|H_{int} |v\ra.
\]
Here $\calM_n[s]=\{ u\in \calM_n\, : \,  \mathrm{Dyck}(u)=s\}$.
Note that $\la u|H_{int} |v\ra=-1/2$ if $u$ and $v$ differ exactly at two consecutive positions
where $u$ and $v$ contain $00$ and $lr$ respectively or vice versa.
In all other cases one has $\la u|H_{int} |v\ra=0$.

Suppose $t\in \calD_{k+1}$ and $P(s,t)>0$.
Let us fix some $j\in [0,2k]$ such that
$t$ can be obtained from $s$ by inserting a pair $lr$ between $s_j$ and $s_{j+1}$.
For any string $u\in \calM_n[s]$ in which $s_j$ and $s_{j+1}$ are separated by at least two
zeros one can find at least one $v\in \calM_n[t]$ such that
$\la u|H_{int} |v\ra=-1/2$. The fraction of strings $u\in \calM_n[s]$
in which $s_j$ and $s_{j+1}$ are separated by two or more
zeros is at least $1/n^2$ which implies
\[
P(s,t)\ge \frac1{2n^3}.
\]
Suppose now that $t\in \calD_{k-1}$.
Let us fix some $j\in [1,2k-1]$ such that $t$ can be obtained from $s$
by removing the pair $s_j s_{j+1}=lr$. For any string $u\in \calM_n[s]$ in which $s_j$ and $s_{j+1}$
are not separated by zeros one can find at least one $v\in \calM_n[t]$ such that
$\la u|H_{int} |v\ra=-1/2$. The fraction of strings $u\in \calM_n[s]$
in which $s_j$ and $s_{j+1}$ are not separated by
zeros is at least $1/n$ which implies
\[
P(s,t)\ge \frac1{2n^2}.
\]

\subsection{Proof of Lemma~\ref{lemma:map}}

Let us first prove a simple result concerning fractional matchings.
Consider a bipartite graph $G=(A\cup B,E)$.
Let $x=\{x_e\}_{e\in E}$ be a vector of real variables
associated with edges of the graph. For any vertex $u$ let $\delta(u)$ be the set
of edges incident to $u$.
Define a {\em matching polytope}~\cite{Schrijver}
%SBB4: definition of the polytope is changed - a lower bound for A-vertices is added
%the old proof applies to the new definition without changes; please double check this.
\bea
\calP&=& \{ x\, : \, x_e\ge 0 \quad \mbox{for all $e\in E$}, \nn \\
&& 1\le \sum_{e\in \delta(a)} x_e \le 4,  \quad
\sum_{e\in \delta(b)} x_e=1 \nn \\
&& \mbox{for all $a\in A$ and $b\in B$}\}. \nn
\eea
\begin{lemma}
\label{lemma:matching}
Suppose $\calP$ is non-empty. Then there exists a map $f\, : \, B\to A$ such that
(i) $f(b)=a$ implies $(a,b)\in E$, (ii) any vertex $a\in A$ has at least one pre-image in $B$,
and (iii) any vertex $a\in A$ has at most four pre-images in $B$.
\end{lemma}
\begin{proof}
Since $\calP$ is non-empty, it must have at least one extremal point $x^*\in \calP$.
Let $E^*\subseteq E$ be the set of edges such that $x^*_e>0$. We claim that $E^*$ is a forest
(disjoint union of trees). Indeed,
suppose $E^*$ contains a cycle $Z$ (a closed path). Then $x^*_{a,b}<1$ for
all $(a,b)\in Z$ since otherwise the cycle would terminate at $b$.
Hence $0<x^*_e<1$ for all $e\in Z$.
Since the graph is bipartite, one can label edges of $Z$ as even and odd in alternating order.
There exists $\epsilon\ne 0$ such that $x^*$ can
be shifted by $\pm \epsilon$ on even and odd edges of $Z$  respectively without leaving
$\calP$. Hence $x^*$ is a convex mixture of two distinct vectors from $\calP$.
This is a  contradiction since $x^*$ is an extreme point. Hence $E^*$ contains no cycles, that is,
$E^*$ is a forest. We claim that $x_e^*=1$ for all $e\in E^*$. Indeed, let $T\subseteq E^*$
be the subset of edges with $0<x_e^*<1$. Obviously, $T$ itself is a forest. Degree-$1$ nodes
of $T$ must be in $A$ and there must exist a path $\gamma\subseteq T$ starting
and ending at degree-$1$ nodes $u,u'\in A$. Since $0<x^*_e<1$ for all $e\in \gamma$,
there exists $\epsilon\ne 0$ such that $x^*$ can be shifted by $\pm \epsilon$ on
even and odd edges of $\gamma$ respectively without leaving $\calP$.
This is a  contradiction since $x^*$ is an extreme point. Hence $x_e^*=1$ for all $e\in E^*$.
We conclude that $x^*_e\in \{0,1\}$ for all edges of $G$.
The desired map can now be defined as $f(b)=a$ iff $x^*_{a,b}=1$.
\end{proof}

We are interested in the case where
\[
A=\calD_{n-1} \quad \mbox{and} \quad  B=\calD_n
\]
are Dyck paths of semilength $n-1$ and $n$ respectively. Paths $a\in \calD_{n-1}$ and $b\in \calD_n$ are
connected by an edge iff $a$ can be obtained from $b$ by removing a single $lr$ pair.
Our goal is to construct a map $f\,: \, \calD_n\to \calD_{n-1}$
with the properties (i),(ii),(iii) stated in Lemma~\ref{lemma:matching}.
According to the lemma, it suffices to choose
$f$ as a stochastic map. Namely, for any $b\in \calD_n$ we shall define a random
variable $f(b)\in \calD_{n-1}$ with some normalized probability distribution.
It suffices  to satisfy two conditions:
\be
\label{goal1}
\mathrm{Pr}[f(b)=a]>0 \quad \parbox[t]{5cm}{only if $a$ can be obtained from $b$ by removing a single $lr$ pair,}
\ee
and
%SBB4: here and below all inequalities are changed to equalities; lower bound on X_n is added
\be
\label{goal2}
\sum_{b\in \calD_n} \mathrm{Pr}[f(b)=a]=X_n \quad \quad \mbox{for all $a\in \calD_{n-1}$}.
\ee
Here $1\le X_n\le 4$ is some function of $n$ that we shall choose later.
We shall define $f$ using induction in $n$.
\begin{lemma}
\label{lemma:small}
For every $n\ge 1$ there exists $1\le X_n\le 4$ and a stochastic map $f\, : \, \calD_n \to \calD_{n-1}$
satisfying Eqs.~(\ref{goal1},\ref{goal2}).
\end{lemma}
\begin{proof}
Any Dyck path $b\in \calD_n$ can be uniquely represented as $b=lsrt$
for some $s\in \calD_i$, $t\in \calD_{n-i-1}$, and $i\in [0,n-1]$.
We shall specify the map $f\, : \, \calD_n\to \calD_{n-1}$ by the following rules:
\begin{center}
\begin{tabular}{c|c|c}
$b\in \calD_n$ & $f(b)\in \calD_{n-1}$ & probability \\
\hline
$lsrt$, $s\in \calD_i$, $1\le i\le n-2$ & $lf(s) r t$ & $p_i$ \\
\hline
$lsrt$, $s\in \calD_i$, $1\le i\le n-2$ & $lsr f(t)$ & $1-p_i$ \\
\hline
$lrt$, $t\in \calD_{n-1}$ & $lrf(t)$ & $1$ \\
\hline
$lsr$, $s\in \calD_{n-1}$ & $lf(s)r$ & $1$ \\
\hline
\end{tabular}
\end{center}
Here we assumed that $f$ has been already defined for strings of semilength up to $n-1$
such that Eqs.~(\ref{goal1},\ref{goal2}) are satisfied.
By abuse of notation, we ignore the index $n$ in $f$, so we regard $f$ as a family of maps defined for all $n$.
It is clear that our inductive definition of $f$ on $\calD_n$ satisfies Eq.~(\ref{goal1}).
The probabilities $p_1,\ldots,p_{n-2}\in [0,1]$ are free parameters that must be chosen to satisfy
Eq.~(\ref{goal2}). Note that these probabilities also implicitly depend on $n$.
%SBB3: new sentence
The  choices of $f(b)$ in the first two lines of the above table are represented by
black and red nodes in the example shown on Fig.~\ref{fig:supertree}.
Consider three cases:

{\em Case~1:} $a=lrt'$ for some $t'\in \calD_{n-2}$. Then $f(b)=a$ iff
$b=llrrt'$ or $b=lrt$ for some $t\in \calD_{n-1}$ such that $f(t)=t'$.
These possibilities are mutually exclusive.
Hence
\bea
 \sum_{b\in \calD_n} \mathrm{Pr}[f(b)=lrt'] &=& p_1 + \sum_{t\in \calD_{n-1}} \; \mathrm{Pr}[f(t)=t'] \nn   \\
&=&   p_1+X_{n-1}. \nn
\eea
Substituting it into Eq.~(\ref{goal2}) gives a constraint
\be
\label{C1}
p_1= X_n- X_{n-1}.
\ee

{\em Case~2:} $a=ls'r$ for some $s'\in \calD_{n-2}$. Then $f(b)=a$ iff
$b=ls'rlr$ or $b=lsr$ for some $s\in \calD_{n-1}$ such that $f(s)=s'$.
These possibilities are mutually exclusive. Hence
\bea
\sum_{b\in \calD_n} \mathrm{Pr}[f(b)=ls'r] &=&  1-p_{n-2} +
\sum_{s\in \calD_{n-1}} \mathrm{Pr}[f(s)=s']    \nn \\
&= &   1- p_{n-2} + X_{n-1}.  \nn
\eea
Substituting it into Eq.~(\ref{goal2}) gives a constraint
\be
\label{C2}
p_{n-2}=1 -(X_n-X_{n-1}).
\ee
It says that $X_n$ must be a non-decreasing sequence.

{\em Case~3:} $a=ls'rt'$ for some $s'\in \calD_i$, $t'\in \calD_{n-i-2}$, and $i=1,\ldots,n-3$.
In other words, both $s'$ and $t'$ must be non-empty. Then $f(b)=a$ iff
$b=lsrt'$ for some $s\in \calD_{i+1}$ such that $f(s)=s'$,
or $b=ls'rt$ for some $t\in \calD_{n-i-1}$ such that $f(t)=t'$.
These possibilities are mutually exclusive. Hence
\bea
\sum_{b\in \calD_n} \mathrm{Pr}[f(b)=ls'rt'] &=&
p_{i+1}\sum_{s\in \calD_{i+1}} \mathrm{Pr}[f(s)=s'] \nn \\
&& + (1-p_i) \sum_{t\in \calD_{n-i-1}} \mathrm{Pr}[f(t)=t'] \nn \\
&= &   p_{i+1}\, X_{i+1}  +  (1-p_i) X_{n-i-1}.\nn
\eea
Substituting it into Eq.~(\ref{goal2}) gives a constraint
\be
\label{C3}
 p_{i+1} \, X_{i+1}  +  (1-p_i)  X_{n-i-1} =X_n
\ee
for each $i=1,\ldots,n-3$.
Let us choose
\be
\label{Xi}
X_i=\frac{C_i}{C_{i-1}} =\frac{4(i-1/2)}{i+1}.
\ee
Combining Eqs.~(\ref{C1},\ref{C2},\ref{C3})
we obtain a linear system with unknown variables $p_1,\ldots,p_{n-2}\in [0,1]$.
We shall look for a solution $\{p_i\}$ having an extra symmetry
\be
\label{symmetry}
p_i+p_{n-i-1}=1 \quad \mbox{for $i=1,\ldots,n-2$}.
\ee
One can check that the system defined by Eqs.~(\ref{C1},\ref{C2},\ref{C3},\ref{symmetry})
has a solution
\be
p_i=\frac{i(i+1)(3n-2i-1)}{n(n+1)(n-1)}, \quad i=1,\ldots,n-2.
\ee
Hence we have defined the desired stochastic map $f\, :\, \calD_n\to \calD_{n-1}$.
This proves the induction hypothesis.

%SBB3: we also have to consider the case n=1 because it is implicitly used in the inductive definition of f
It remains to note that for $n=1,2$ the map $f$ is uniquely specified by Eqs.~(\ref{goal1},\ref{goal2})
and our choice of $X_n$. Indeed, one has $\calD_2=\{llrr,lrlr\}$, $\calD_1=\{lr\}$, and $\calD_0=\emptyset$.
To satisfy Eq.~(\ref{goal1}),  we have to choose $f(llrr)=f(lrlr)=lr$ for $n=2$ and $f(lr)=\emptyset$ for $n=1$.
It also satisfies Eq.~(\ref{goal2}) since $X_2=2$ and $X_1=1$.
This proves the base of induction.
\end{proof}

\subsection{Ground state energy for unbalanced subspaces}

Recall that the unbalanced subspace $\calH_{p,q}$ is spanned by strings
$s\in C_{p,q}$ that have $p$ unmatched right and $q$ unmatched left brackets.
Our goal is to prove that the restriction of $H$ onto any subspace $\calH_{p,q}$
with $p>0$ or $q>0$ has ground state energy at least $n^{-O(1)}$.
By the symmetry, it suffices to consider the case $p>0$.
To simplify the analysis we shall omit the boundary term $|l\ra\la l|_n$.
Note that such omission can only decrease the ground state energy. Accordingly, our
simplified  Hamiltonian becomes
\be
\label{Hsimplified}
H=|r\ra \la r|_1 + \sum_{j=1}^{n-1} \Pi_{j,j+1}.
\ee
Recall that $\Pi$ is a projector  onto the subspace spanned by states
$|00\ra -|lr\ra$, $|0l\ra-|l0\ra$, and $|0r\ra-|r0\ra$.
Let $\lambda_1(H)$ be the ground state energy of $H$.

Any string $s\in C_{p,q}$ can be uniquely represented as
\[
s=u_0 r u_1 r u_2 \ldots r u_p l v_1 l v_2 \ldots l v_q
\]
where $u_i$ and $v_j$ are Motzkin paths (balanced strings of brackets).
The remaining $p$ right and $q$ left brackets are unmatched
and never participate in the move $00\leftrightarrow lr$.
It follows that the  unmatched brackets can be regarded
as ``solid walls'' that can be swapped with $0$'s but otherwise
do not participate in any interactions. In particular, the
spectrum of $H$ restricted to $\calH_{p,q}$ depends only on $p+q$ as long as $p>0$.
This allows us to focus on the case $q=0$, i.e.
assume that all unmatched brackets are right.

Given a string $s\in C_{p,0}$, let $\tilde{s}\in  \{0,l,r,x,y\}^n$ be the string
obtained from $s$ by the following operations: (i) replace the first unmatched right
bracket in $s$ by `$x$', and (ii) replace all other unmatched brackets in $s$ (if any) by `$y$'.
Define a new Hilbert space $\tilde{\calH}_p$ whose basis vectors are $|\tilde{s}\ra$, $s\in C_{p,0}$.
Consider a Hamiltonian
\be
\label{Htilde}
\tilde{H}=|x\ra\la x|_1 + \sum_{j=1}^{n-1} \Pi_{j,j+1} + \Theta^x_{j,j+1} + \Theta^y_{j,j+1},
\ee
where
$\Theta^x$ and $\Theta^y$ are  projectors onto the states
$|0x\ra-|x0\ra$ and $|0y\ra-|y0\ra$ respectively (with a proper normalization).
One can easily check that $\la s|H|t\ra =\la \tilde{s}|\tilde{H}|\tilde{t}\ra$ for any $s,t\in C_{p,0}$.
Hence the spectrum of $H$ on $\calH_{p,0}$ coincides with the spectrum of $\tilde{H}$
on $\tilde{\calH}_p$. Furthermore, if we omit all the terms $\Theta^y_{j,j+1}$ in $\tilde{H}$, the ground state energy
can only decrease. Hence it suffices to consider a simplified Hamiltonian
\be
\label{Hx}
H^x=|x\ra\la x|_1 + \sum_{j=1}^{n-1} \Pi_{j,j+1} + \Theta^x_{j,j+1}
\ee
which acts on $\tilde{\calH}_p$.
Note that positions of $y$-particles are integrals of motion for $H^x$.
Moreover, for fixed positions of $y$-particles,
any term in $H^x$ touching a $y$-particle vanishes. Hence
$H^x$ can be analyzed separately on each  interval between consecutive $y$-particles.
Since our goal is to get a lower bound on the ground state energy,
we can only analyze the interval between $1$ and the first $y$-particle.
Equivalently, we can redefine $n$ and focus on the case $p=1$, $q=0$, that is,
assume that there is only one unmatched right bracket.
The relevant Hilbert space $\tilde{\calH}_1$ is now spanned by states
\[
|s\ra \otimes |x\ra \otimes |t\ra, \quad \mbox{where} \quad s\in \calM_{j-1}, \quad t\in \calM_{n-j}.
\]
Recall that $\calM_k$ is the set of Motzkin paths (balanced strings of left and right brackets)
of length $k$.

We would like to treat the terms responsible for the motion and detection of
the $x$-particle as a small perturbation. To this end, choose any $0<\epsilon\le 1$ and define
the Hamiltonian
\[
H^x_\epsilon = \sum_{j=1}^{n-1} \Pi_{j,j+1} + \epsilon |x\ra\la x|_1 + \epsilon \sum_{j=1}^{n-1} \Theta^x_{j,j+1}.
\]
Clearly, $H^x_\epsilon\le H^x$, so it suffices to get a lower bound
on the ground state energy of $H^x_\epsilon$.

Let us first find the ground subspace and the spectral gap of the unperturbed Hamiltonian $H^x_0=\sum_{j=1}^{n-1} \Pi_{j,j+1}$.
Note that the position of the $x$-particle $j$ is an invariant of motion for $H^x_0$.
Moreover, any projector $\Pi_{i,i+1}$ touching the $x$-particle vanishes.
Hence we can analyze $H^x_0$ separately on the two disjoint intervals $A=[1,j-1]$ and
$B=[j+1,n]$. It follows that the ground subspace of $H^x_0$ is spanned by normalized states
\be
\label{fixed_j}
|\psi_j\ra = |\calM_{j-1}\ra \otimes |x\ra \otimes |\calM_{n-j}\ra, \quad j=1,\ldots,n.
\ee
The spectral gap of  $H^x_0$ can also be computed separately in $A$ and $B$.
Since we have already shown that the original Hamiltonian
Eq.~(\ref{H}) has a polynomial gap inside the Motzkin subspace, we conclude that
$\lambda_2(H^x_0)\ge n^{-O(1)}$.

Let us now turn on the perturbation.
The first-order effective Hamiltonian acting on the ground subspace
spanned by $\psi_1,\ldots,\psi_n$ describes a hopping of the $x$-particle
on a chain of length $n$ with a delta-like repulsive potential applied at site $j=1$.
Parameters of the hopping Hamiltonian can be found by calculating the matrix elements
\[
\la \psi_j|\Theta^x_{j,j+1} |\psi_j\ra =\frac{M_{n-j-1}}{2M_{n-j}}\equiv \alpha_j^2,
\]
\[
\la \psi_{j+1}|\Theta^x_{j,j+1} |\psi_{j+1}\ra =\frac{M_{j-1}}{2M_{j}}\equiv \beta_j^2,
\]
and
\[
\la \psi_j| \Theta^x_{j,j+1} |\psi_{j+1}\ra =-\frac12 \sqrt{\frac{M_{n-j-1}}{M_{n-j}} \cdot \frac{M_{j-1}}{M_{j}}}=-\alpha_j \beta_j,
\]
where $M_k=|\calM_k|$ is the $k$-th Motzkin number.
We arrive at the effective hopping Hamiltonian acting on $\CC^n$, namely,
\be
\label{Heff1D}
H_{\eff}=|1\ra\la 1| + \sum_{j=1}^{n-1} \Gamma_{j,j+1},
\ee
where
\bea
\Gamma_{j,j+1} &=& \alpha_j^2 \, |j\ra\la j| + \beta_j^2\,  |j+1\ra\la j+1| \nn \\
&& - \alpha_j\beta_j (|j\ra\la j+1| + |j+1\ra\la j|)
\eea
is a rank-$1$ projector.
Applying the Projection Lemma of~\cite{KKR04} we infer that
%SBB4: the same mistake as above (wrong power of \|V\| )
\[
\lambda_1(H^x_\epsilon)\ge \epsilon \lambda_1(H_{\eff}) - \frac{O(\epsilon^2) \| V\|^2}{\lambda_2(H^x_0) - 2\epsilon \|V\|},
\]
where $V=|x\ra\la x|_1+\sum_{j=1}^{n-1} \Theta^x_{j,j+1}$ is the perturbation operator. Since $\lambda_2(H^x_0) \ge n^{-O(1)}$,
we can choose $\epsilon$ polynomial in $1/n$ such that $2\epsilon \|V\|$ is small compared with $\lambda_2(H^x_0)$.
For this choice of $\epsilon$ one gets
\[
\lambda_1(H^x_\epsilon)\ge \epsilon \lambda_1(H_{\eff}) - O(\epsilon^2) n^{O(1)}.
\]
Hence it suffices to show that $\lambda_1(H_{\eff})\ge n^{-O(1)}$,
where $H_{\eff}$ is now the single $x$-particle hopping Hamiltonian Eq.~(\ref{Heff1D}).

Let us first focus on the hopping Hamiltonian without the repulsive potential:
\[
H_{move}\equiv \sum_{j=1}^{n-1} \Gamma_{j,j+1}.
\]
This Hamiltonian is FF and its unique ground state is
\be
\label{g}
|g\ra \sim \sum_{j=1}^n \sqrt{M_{j-1} M_{n-j}} \, |j\ra.
\ee
Our strategy will be to bound the spectral gap of $H_{move}$
and apply the Projection Lemma to $H_{\eff}$ by treating the repulsive potential $|1\ra\la 1|$
as a perturbation of $H_{move}$. First let us
map $H_{move}$ to a stochastic matrix describing a random
walk on the interval $[1,n]$ with the steady state $\pi(j)=\la j|g\ra^2$.
For any $a,b\in [1,n]$ define
\be
\label{walk2}
P(j,k)=\delta_{j,k} - \la j|H_{move}|k\ra \sqrt{\frac{\pi(k)}{\pi(j)}}.
\ee
Since $\sqrt{\pi(j)}$ is a zero eigenvector of $H_{move}$, we
infer that $\sum_k P(j,k)=1$ and $\sum_j \pi(j) P(j,k)=\pi(k)$.
A simple algebra shows that
\[
P(j,j+1)=\frac{M_{n-j-1}}{2M_{n-j}} \quad \mbox{and} \quad P(j+1,j)=\frac{M_{j-1}}{2M_{j}}
\]
are the only non-zero off-diagonal matrix elements of $P$.
We shall use the following property of the Motzkin numbers.
\begin{lemma}
For any $n\ge 1$ one has $1/3 \le M_n/M_{n+1}\le 1$.
Furthermore, for large $n$ one can use an approximation
\be
M_n \approx c\frac{3^n}{n^{3/2}}
\ee
where $c\approx 1.46$.
\end{lemma}
The lemma implies that
\[
\frac16 \le P(j,j\pm 1)\le \frac12
\]
for all $j$. Hence the diagonal matrix elements $P(j,j)$ are non-negative, that is,
we indeed can regard $P(j,k)$ as a transition probability from $j$ to $k$.
Furthermore, using Eq.~(\ref{g}) and the above lemma we infer that the steady state $\pi$ is `almost uniform', that is,
\be
n^{-O(1)} \le \frac{\pi(k)}{\pi(j)} \le n^{O(1)} \quad \mbox{for all $1\le j,k\le n$}.
\ee
In particular, $\min_j{\pi(j)} \ge n^{-O(1)}$.
We can now easily bound the spectral gap of $P$.  For example, applying the canonical
paths theorem stated above we get $1-\lambda_2(P)\ge 1/(\rho l)$
where $\rho$ is defined in Eq.~(\ref{edge_load}) and the canonical path $\gamma(s,t)$ simply
moves the $x$-particle from $s$ to $t$. Since the denominator in Eq.~(\ref{edge_load})
is lower bounded by $n^{-O(1)}$, we conclude that $1-\lambda_2(P)\ge n^{-O(1)}$.
It shows that $\lambda_2(H_{move}) \ge n^{-O(1)}$.

To conclude the proof, it remains to apply the Projection Lemma to
$H_{\eff}$ defined in Eq.~(\ref{Heff1D}) by treating the
repulsive potential $|1\ra\la 1|$ as  a perturbation.
Now the effective first-order Hamiltonian will be simply a $c$-number $\la 1|g\ra^2=\pi(1)\ge n^{-O(1)}$
which proves the bound $\lambda_1(H_{\eff}) \ge n^{-O(1)}$.

\section{Acknowledgments}
We thank Alexei Kitaev, Joel B. Lewis,  Richard P. Stanley, and Guifre Vidal  for useful discussions.
SB was partially  supported by the
DARPA QUEST program under contract number HR0011-09-C-0047.
DN and LC acknowledge support from the European project
Q-ESSENCE 2010-248095 (7th FP), the Slovak Research and Development
Agency under the contract No. LPP-0430-09, and COQI APVV-0646-10.
RM and PS were supported in
part by the National Science Foundation
through grant number CCF-0829421, and PS was supported in part by the
U.S. Army Research Office through grant number W911NF-09-1-0438.

%\bibliography{mybib}

\begin{thebibliography}{30}%
\makeatletter
\providecommand \@ifxundefined [1]{%
 \@ifx{#1\undefined}
}%
\providecommand \@ifnum [1]{%
 \ifnum #1\expandafter \@firstoftwo
 \else \expandafter \@secondoftwo
 \fi
}%
\providecommand \@ifx [1]{%
 \ifx #1\expandafter \@firstoftwo
 \else \expandafter \@secondoftwo
 \fi
}%
\providecommand \natexlab [1]{#1}%
\providecommand \enquote  [1]{``#1''}%
\providecommand \bibnamefont  [1]{#1}%
\providecommand \bibfnamefont [1]{#1}%
\providecommand \citenamefont [1]{#1}%
\providecommand \href@noop [0]{\@secondoftwo}%
\providecommand \href [0]{\begingroup \@sanitize@url \@href}%
\providecommand \@href[1]{\@@startlink{#1}\@@href}%
\providecommand \@@href[1]{\endgroup#1\@@endlink}%
\providecommand \@sanitize@url [0]{\catcode `\\12\catcode `\$12\catcode
  `\&12\catcode `\#12\catcode `\^12\catcode `\_12\catcode `\%12\relax}%
\providecommand \@@startlink[1]{}%
\providecommand \@@endlink[0]{}%
\providecommand \url  [0]{\begingroup\@sanitize@url \@url }%
\providecommand \@url [1]{\endgroup\@href {#1}{\urlprefix }}%
\providecommand \urlprefix  [0]{URL }%
\providecommand \Eprint [0]{\href }%
\@ifxundefined \urlstyle {%
  \providecommand \doi  [0]{\begingroup \@sanitize@url \@doi}%
  \providecommand \@doi [1]{\endgroup \@@startlink {\doibase
  #1}doi:\discretionary {}{}{}#1\@@endlink }%
}{%
  \providecommand \doi  [0]{doi:\discretionary{}{}{}\begingroup
  \urlstyle{rm}\Url }%
}%
\providecommand \doibase [0]{http://dx.doi.org/}%
\providecommand \Doi [0]{\begingroup \@sanitize@url \@Doi }%
\providecommand \@Doi  [1]{\endgroup\@@startlink{\doibase#1}\@@Doi}%
\providecommand \@@Doi [1]{#1\@@endlink}%
\providecommand \selectlanguage [0]{\@gobble}%
\providecommand \bibinfo  [0]{\@secondoftwo}%
\providecommand \bibfield  [0]{\@secondoftwo}%
\providecommand \translation [1]{[#1]}%
\providecommand \BibitemOpen [0]{}%
\providecommand \bibitemStop [0]{}%
\providecommand \bibitemNoStop [0]{.\EOS\space}%
\providecommand \EOS [0]{\spacefactor3000\relax}%
\providecommand \BibitemShut  [1]{\csname bibitem#1\endcsname}%
%</preamble>
\bibitem [{\citenamefont {Sachdev}(2011)}]{Sachdev11}%
  \BibitemOpen
  \bibfield  {author} {\bibinfo {author} {\bibfnamefont {S.}~\bibnamefont
  {Sachdev}},\ }\href@noop {} {\emph {\bibinfo {title} {Quantum Phase
  Transitions}}}\ (\bibinfo  {publisher} {Cambridge University Press},\
  \bibinfo {year} {2011})\BibitemShut {NoStop}%
\bibitem [{\citenamefont {Vidal}\ \emph {et~al.}(2003)\citenamefont {Vidal},
  \citenamefont {Latorre}, \citenamefont {Rico},\ and\ \citenamefont
  {Kitaev}}]{Vidal03}%
  \BibitemOpen
  \bibfield  {author} {\bibinfo {author} {\bibfnamefont {G.}~\bibnamefont
  {Vidal}}, \bibinfo {author} {\bibfnamefont {J.~I.}\ \bibnamefont {Latorre}},
  \bibinfo {author} {\bibfnamefont {E.}~\bibnamefont {Rico}}, \ and\ \bibinfo
  {author} {\bibfnamefont {A.}~\bibnamefont {Kitaev}},\ }\href@noop {}
  {\bibfield  {journal} {\bibinfo  {journal} {Phys.~Rev.~Lett},\ }\textbf
  {\bibinfo {volume} {90}},\ \bibinfo {pages} {227902} (\bibinfo {year}
  {2003})}\BibitemShut {NoStop}%
\bibitem [{\citenamefont {Korepin}(2004)}]{Korepin04}%
  \BibitemOpen
  \bibfield  {author} {\bibinfo {author} {\bibfnamefont {V.~E.}\ \bibnamefont
  {Korepin}},\ }\href@noop {} {\bibfield  {journal} {\bibinfo  {journal} {Phys.
  Rev. Lett.},\ }\textbf {\bibinfo {volume} {92}},\ \bibinfo {pages} {096402}
  (\bibinfo {year} {2004})}\BibitemShut {NoStop}%
\bibitem [{\citenamefont {Hastings}(2007)}]{Hastings07}%
  \BibitemOpen
  \bibfield  {author} {\bibinfo {author} {\bibfnamefont {M.~B.}\ \bibnamefont
  {Hastings}},\ }\href@noop {} {\bibfield  {journal} {\bibinfo  {journal} {J.
  Stat. Mech.},\ \bibinfo {pages} {P08024}} (\bibinfo {year}
  {2007})}\BibitemShut {NoStop}%
\bibitem [{\citenamefont {Eisert}\ \emph {et~al.}(2010)\citenamefont {Eisert},
  \citenamefont {Cramer},\ and\ \citenamefont {Plenio}}]{Eisert08}%
  \BibitemOpen
  \bibfield  {author} {\bibinfo {author} {\bibfnamefont {J.}~\bibnamefont
  {Eisert}}, \bibinfo {author} {\bibfnamefont {M.}~\bibnamefont {Cramer}}, \
  and\ \bibinfo {author} {\bibfnamefont {M.}~\bibnamefont {Plenio}},\
  }\href@noop {} {\bibfield  {journal} {\bibinfo  {journal} {Rev. Mod. Phys.},\
  }\textbf {\bibinfo {volume} {82}},\ \bibinfo {pages} {277} (\bibinfo {year}
  {2010})}\BibitemShut {NoStop}%
\bibitem [{\citenamefont {Arad}\ \emph {et~al.}(2011)\citenamefont {Arad},
  \citenamefont {Landau},\ and\ \citenamefont {Vazirani}}]{Arad11}%
  \BibitemOpen
  \bibfield  {author} {\bibinfo {author} {\bibfnamefont {I.}~\bibnamefont
  {Arad}}, \bibinfo {author} {\bibfnamefont {Z.}~\bibnamefont {Landau}}, \ and\
  \bibinfo {author} {\bibfnamefont {U.}~\bibnamefont {Vazirani}},\ }\href@noop
  {} { (\bibinfo {year} {2011})},\ \Eprint
  {http://arxiv.org/abs/arXiv:1111.2970} {arXiv:1111.2970} \BibitemShut
  {NoStop}%
\bibitem [{\citenamefont {Latorre}\ \emph {et~al.}(2004)\citenamefont
  {Latorre}, \citenamefont {Rico},\ and\ \citenamefont {Vidal}}]{Latorre03}%
  \BibitemOpen
  \bibfield  {author} {\bibinfo {author} {\bibfnamefont {J.~I.}\ \bibnamefont
  {Latorre}}, \bibinfo {author} {\bibfnamefont {E.}~\bibnamefont {Rico}}, \
  and\ \bibinfo {author} {\bibfnamefont {G.}~\bibnamefont {Vidal}},\
  }\href@noop {} {\bibfield  {journal} {\bibinfo  {journal} {Quant. Inf.
  Comput.},\ }\textbf {\bibinfo {volume} {4}},\ \bibinfo {pages} {48} (\bibinfo
  {year} {2004})}\BibitemShut {NoStop}%
\bibitem [{\citenamefont {Koma}\ and\ \citenamefont
  {Nachtergaele}(1997)}]{Koma95}%
  \BibitemOpen
  \bibfield  {author} {\bibinfo {author} {\bibfnamefont {T.}~\bibnamefont
  {Koma}}\ and\ \bibinfo {author} {\bibfnamefont {B.}~\bibnamefont
  {Nachtergaele}},\ }\href@noop {} {\bibfield  {journal} {\bibinfo  {journal}
  {Lett. Math. Phys.},\ }\textbf {\bibinfo {volume} {40}},\ \bibinfo {pages}
  {1} (\bibinfo {year} {1997})}\BibitemShut {NoStop}%
\bibitem [{\citenamefont {Affleck}\ \emph {et~al.}(1987)\citenamefont
  {Affleck}, \citenamefont {Kennedy}, \citenamefont {Lieb},\ and\ \citenamefont
  {Tasaki}}]{AKLT87}%
  \BibitemOpen
  \bibfield  {author} {\bibinfo {author} {\bibfnamefont {I.}~\bibnamefont
  {Affleck}}, \bibinfo {author} {\bibfnamefont {T.}~\bibnamefont {Kennedy}},
  \bibinfo {author} {\bibfnamefont {E.~H.}\ \bibnamefont {Lieb}}, \ and\
  \bibinfo {author} {\bibfnamefont {H.}~\bibnamefont {Tasaki}},\ }\href@noop {}
  {\bibfield  {journal} {\bibinfo  {journal} {Phys.~Rev.~Lett},\ }\textbf
  {\bibinfo {volume} {59}},\ \bibinfo {pages} {799–802} (\bibinfo {year}
  {1987})}\BibitemShut {NoStop}%
\bibitem [{\citenamefont {Fannes}\ \emph {et~al.}(1992)\citenamefont {Fannes},
  \citenamefont {Nachtergaele},\ and\ \citenamefont {Werner}}]{Fannes92}%
  \BibitemOpen
  \bibfield  {author} {\bibinfo {author} {\bibfnamefont {M.}~\bibnamefont
  {Fannes}}, \bibinfo {author} {\bibfnamefont {B.}~\bibnamefont
  {Nachtergaele}}, \ and\ \bibinfo {author} {\bibfnamefont {R.~F.}\
  \bibnamefont {Werner}},\ }\href@noop {} {\bibfield  {journal} {\bibinfo
  {journal} {Comm. Math. Phys.},\ }\textbf {\bibinfo {volume} {144}},\ \bibinfo
  {pages} {443} (\bibinfo {year} {1992})}\BibitemShut {NoStop}%
\bibitem [{\citenamefont {Perez-Garcia}\ \emph {et~al.}(2007)\citenamefont
  {Perez-Garcia}, \citenamefont {Verstraete}, \citenamefont {Cirac},\ and\
  \citenamefont {Wolf}}]{PEPS07}%
  \BibitemOpen
  \bibfield  {author} {\bibinfo {author} {\bibfnamefont {D.}~\bibnamefont
  {Perez-Garcia}}, \bibinfo {author} {\bibfnamefont {F.}~\bibnamefont
  {Verstraete}}, \bibinfo {author} {\bibfnamefont {J.~I.}\ \bibnamefont
  {Cirac}}, \ and\ \bibinfo {author} {\bibfnamefont {M.~M.}\ \bibnamefont
  {Wolf}},\ }\href@noop {} {\bibfield  {journal} {\bibinfo  {journal} {Quant.
  Inf. Comp.},\ }\textbf {\bibinfo {volume} {8}},\ \bibinfo {pages} {0650}
  (\bibinfo {year} {2007})}\BibitemShut {NoStop}%
\bibitem [{\citenamefont {Chen}\ \emph {et~al.}(2010)\citenamefont {Chen},
  \citenamefont {Chen}, \citenamefont {Duan}, \citenamefont {Ji},\ and\
  \citenamefont {Zeng}}]{Chen04}%
  \BibitemOpen
  \bibfield  {author} {\bibinfo {author} {\bibfnamefont {J.}~\bibnamefont
  {Chen}}, \bibinfo {author} {\bibfnamefont {X.}~\bibnamefont {Chen}}, \bibinfo
  {author} {\bibfnamefont {R.}~\bibnamefont {Duan}}, \bibinfo {author}
  {\bibfnamefont {Z.}~\bibnamefont {Ji}}, \ and\ \bibinfo {author}
  {\bibfnamefont {B.}~\bibnamefont {Zeng}},\ }\href@noop {} { (\bibinfo {year}
  {2010})},\ \Eprint {http://arxiv.org/abs/arXiv:1004.3787} {arXiv:1004.3787}
  \BibitemShut {NoStop}%
\bibitem [{\citenamefont {Kraus}\ \emph {et~al.}(2008)\citenamefont {Kraus},
  \citenamefont {B\"uchler}, \citenamefont {Diehl}, \citenamefont {Kantian},
  \citenamefont {Micheli},\ and\ \citenamefont {Zoller}}]{Kraus08}%
  \BibitemOpen
  \bibfield  {author} {\bibinfo {author} {\bibfnamefont {B.}~\bibnamefont
  {Kraus}}, \bibinfo {author} {\bibfnamefont {H.~P.}\ \bibnamefont
  {B\"uchler}}, \bibinfo {author} {\bibfnamefont {S.}~\bibnamefont {Diehl}},
  \bibinfo {author} {\bibfnamefont {A.}~\bibnamefont {Kantian}}, \bibinfo
  {author} {\bibfnamefont {A.}~\bibnamefont {Micheli}}, \ and\ \bibinfo
  {author} {\bibfnamefont {P.}~\bibnamefont {Zoller}},\ }\href@noop {}
  {\bibfield  {journal} {\bibinfo  {journal} {Phys. Rev. A},\ }\textbf
  {\bibinfo {volume} {78}},\ \bibinfo {pages} {042307} (\bibinfo {year}
  {2008})}\BibitemShut {NoStop}%
\bibitem [{\citenamefont {Verstraete}\ \emph {et~al.}(2009)\citenamefont
  {Verstraete}, \citenamefont {Wolf},\ and\ \citenamefont
  {Cirac}}]{Verstraete08}%
  \BibitemOpen
  \bibfield  {author} {\bibinfo {author} {\bibfnamefont {F.}~\bibnamefont
  {Verstraete}}, \bibinfo {author} {\bibfnamefont {M.~M.}\ \bibnamefont
  {Wolf}}, \ and\ \bibinfo {author} {\bibfnamefont {J.~I.}\ \bibnamefont
  {Cirac}},\ }\href@noop {} {\bibfield  {journal} {\bibinfo  {journal} {Nature
  Physics},\ }\textbf {\bibinfo {volume} {5}},\ \bibinfo {pages} {633}
  (\bibinfo {year} {2009})}\BibitemShut {NoStop}%
\bibitem [{\citenamefont {Diehl}\ \emph {et~al.}(2008)\citenamefont {Diehl},
  \citenamefont {Micheli}, \citenamefont {Kantian}, \citenamefont {Kraus},
  \citenamefont {B\"uchler},\ and\ \citenamefont {Zoller}}]{Diehl08}%
  \BibitemOpen
  \bibfield  {author} {\bibinfo {author} {\bibfnamefont {S.}~\bibnamefont
  {Diehl}}, \bibinfo {author} {\bibfnamefont {A.}~\bibnamefont {Micheli}},
  \bibinfo {author} {\bibfnamefont {A.}~\bibnamefont {Kantian}}, \bibinfo
  {author} {\bibfnamefont {B.}~\bibnamefont {Kraus}}, \bibinfo {author}
  {\bibfnamefont {H.}~\bibnamefont {B\"uchler}}, \ and\ \bibinfo {author}
  {\bibfnamefont {P.}~\bibnamefont {Zoller}},\ }\href@noop {} {\bibfield
  {journal} {\bibinfo  {journal} {Nature Physics},\ }\textbf {\bibinfo {volume}
  {4}},\ \bibinfo {pages} {878} (\bibinfo {year} {2008})}\BibitemShut {NoStop}%
\bibitem [{\citenamefont {Donaghey}\ and\ \citenamefont
  {Shapiro}(1977)}]{Motzkin2}%
  \BibitemOpen
  \bibfield  {author} {\bibinfo {author} {\bibfnamefont {R.}~\bibnamefont
  {Donaghey}}\ and\ \bibinfo {author} {\bibfnamefont {L.}~\bibnamefont
  {Shapiro}},\ }\href@noop {} {\bibfield  {journal} {\bibinfo  {journal}
  {Journal of Combinatorial Theory (A)},\ }\textbf {\bibinfo {volume} {23}},\
  \bibinfo {pages} {291} (\bibinfo {year} {1977})}\BibitemShut {NoStop}%
\bibitem [{\citenamefont {Stanley}(1999)}]{Motzkin1}%
  \BibitemOpen
  \bibfield  {author} {\bibinfo {author} {\bibfnamefont {R.~P.}\ \bibnamefont
  {Stanley}},\ }\href@noop {} {\emph {\bibinfo {title} {Enumerative
  Combinatorics, Volume 2}}}\ (\bibinfo  {publisher} {Cambridge University
  Press},\ \bibinfo {year} {1999})\ p.\ \bibinfo {pages} {238}\BibitemShut
  {NoStop}%
\bibitem [{Note1()}]{Note1}%
  \BibitemOpen
  \bibinfo {note} {Here and below the spectral gap of a Hamiltonian means the
  difference between the smallest and the second smallest
  eigenvalue.}\BibitemShut {Stop}%
\bibitem [{Note2()}]{Note2}%
  \BibitemOpen
  \bibinfo {note} {One can regard Dyck paths as a special case of Motzkin paths
  in which no `$0$' symbols are allowed.}\BibitemShut {Stop}%
\bibitem [{\citenamefont {Kempe}\ \emph {et~al.}(2006)\citenamefont {Kempe},
  \citenamefont {Kitaev},\ and\ \citenamefont {Regev}}]{KKR04}%
  \BibitemOpen
  \bibfield  {author} {\bibinfo {author} {\bibfnamefont {J.}~\bibnamefont
  {Kempe}}, \bibinfo {author} {\bibfnamefont {A.}~\bibnamefont {Kitaev}}, \
  and\ \bibinfo {author} {\bibfnamefont {O.}~\bibnamefont {Regev}},\
  }\href@noop {} {\bibfield  {journal} {\bibinfo  {journal} {SIAM J. of
  Comp.},\ }\textbf {\bibinfo {volume} {35}},\ \bibinfo {pages} {1070}
  (\bibinfo {year} {2006})}\BibitemShut {NoStop}%
\bibitem [{\citenamefont {Schrijver}(2002)}]{Schrijver}%
  \BibitemOpen
  \bibfield  {author} {\bibinfo {author} {\bibfnamefont {A.}~\bibnamefont
  {Schrijver}},\ }\href@noop {} {\emph {\bibinfo {title} {Combinatorial
  Optimization}}}\ (\bibinfo  {publisher} {Springer},\ \bibinfo {year}
  {2002})\BibitemShut {NoStop}%
\bibitem [{\citenamefont {Gottesman}\ and\ \citenamefont
  {Hastings}(2009)}]{Gottesman09}%
  \BibitemOpen
  \bibfield  {author} {\bibinfo {author} {\bibfnamefont {D.}~\bibnamefont
  {Gottesman}}\ and\ \bibinfo {author} {\bibfnamefont {M.~B.}\ \bibnamefont
  {Hastings}},\ }\href@noop {} {\bibfield  {journal} {\bibinfo  {journal} {New
  J. Phys.},\ }\textbf {\bibinfo {volume} {12}},\ \bibinfo {pages} {025002}
  (\bibinfo {year} {2009})}\BibitemShut {NoStop}%
\bibitem [{\citenamefont {Irani}(2010)}]{Irani09}%
  \BibitemOpen
  \bibfield  {author} {\bibinfo {author} {\bibfnamefont {S.}~\bibnamefont
  {Irani}},\ }\href@noop {} {\bibfield  {journal} {\bibinfo  {journal} {J.
  Math. Phys.},\ }\textbf {\bibinfo {volume} {51}},\ \bibinfo {pages} {022101}
  (\bibinfo {year} {2010})}\BibitemShut {NoStop}%
\bibitem [{\citenamefont {Movassagh}\ \emph {et~al.}(2010)\citenamefont
  {Movassagh}, \citenamefont {Farhi}, \citenamefont {Goldstone}, \citenamefont
  {Nagaj}, \citenamefont {Osborne},\ and\ \citenamefont {Shor}}]{MFGNOS}%
  \BibitemOpen
  \bibfield  {author} {\bibinfo {author} {\bibfnamefont {R.}~\bibnamefont
  {Movassagh}}, \bibinfo {author} {\bibfnamefont {E.}~\bibnamefont {Farhi}},
  \bibinfo {author} {\bibfnamefont {J.}~\bibnamefont {Goldstone}}, \bibinfo
  {author} {\bibfnamefont {D.}~\bibnamefont {Nagaj}}, \bibinfo {author}
  {\bibfnamefont {T.~J.}\ \bibnamefont {Osborne}}, \ and\ \bibinfo {author}
  {\bibfnamefont {P.~W.}\ \bibnamefont {Shor}},\ }\href@noop {} {\bibfield
  {journal} {\bibinfo  {journal} {Phys.~Rev~{\bf A}},\ }\textbf {\bibinfo
  {volume} {82}},\ \bibinfo {pages} {012318} (\bibinfo {year}
  {2010})}\BibitemShut {NoStop}%
\bibitem [{Note3()}]{Note3}%
  \BibitemOpen
  \bibinfo {note} {Though the results of Ref.~\cite {MFGNOS} are applicable to
  more general Hamiltonians, the convenient restriction to random projectors is
  sufficient for addressing the degeneracy and frustration
  condition.}\BibitemShut {Stop}%
\bibitem [{\citenamefont {Wolf}\ \emph {et~al.}(2006)\citenamefont {Wolf},
  \citenamefont {Ortiz}, \citenamefont {Verstraete},\ and\ \citenamefont
  {Cirac}}]{Wolf06}%
  \BibitemOpen
  \bibfield  {author} {\bibinfo {author} {\bibfnamefont {M.~M.}\ \bibnamefont
  {Wolf}}, \bibinfo {author} {\bibfnamefont {G.}~\bibnamefont {Ortiz}},
  \bibinfo {author} {\bibfnamefont {F.}~\bibnamefont {Verstraete}}, \ and\
  \bibinfo {author} {\bibfnamefont {J.~I.}\ \bibnamefont {Cirac}},\ }\href@noop
  {} {\bibfield  {journal} {\bibinfo  {journal} {Phys. Rev. Lett.},\ }\textbf
  {\bibinfo {volume} {97}},\ \bibinfo {pages} {110403} (\bibinfo {year}
  {2006})}\BibitemShut {NoStop}%
\bibitem [{\citenamefont {Diaconis}\ and\ \citenamefont
  {Stroock}(1991)}]{Diaconis91}%
  \BibitemOpen
  \bibfield  {author} {\bibinfo {author} {\bibfnamefont {P.}~\bibnamefont
  {Diaconis}}\ and\ \bibinfo {author} {\bibfnamefont {D.}~\bibnamefont
  {Stroock}},\ }\href@noop {} {\bibfield  {journal} {\bibinfo  {journal} {Ann.
  Appl. Probab.},\ }\textbf {\bibinfo {volume} {1}},\ \bibinfo {pages} {36}
  (\bibinfo {year} {1991})}\BibitemShut {NoStop}%
\bibitem [{\citenamefont {Sinclair}(1992)}]{Sinclair92}%
  \BibitemOpen
  \bibfield  {author} {\bibinfo {author} {\bibfnamefont {A.}~\bibnamefont
  {Sinclair}},\ }\href@noop {} {\bibfield  {journal} {\bibinfo  {journal}
  {Combinatorics, Probability, and Computing},\ }\textbf {\bibinfo {volume}
  {1}},\ \bibinfo {pages} {351} (\bibinfo {year} {1992})}\BibitemShut {NoStop}%
\bibitem [{\citenamefont {Verstraete}\ and\ \citenamefont
  {Cirac}(2006)}]{Verstraete05}%
  \BibitemOpen
  \bibfield  {author} {\bibinfo {author} {\bibfnamefont {F.}~\bibnamefont
  {Verstraete}}\ and\ \bibinfo {author} {\bibfnamefont {J.}~\bibnamefont
  {Cirac}},\ }\href@noop {} {\bibfield  {journal} {\bibinfo  {journal} {Phys.
  Rev. B},\ }\textbf {\bibinfo {volume} {73}},\ \bibinfo {pages} {094423}
  (\bibinfo {year} {2006})}\BibitemShut {NoStop}%
\bibitem [{\citenamefont {Schuch}\ \emph {et~al.}(2008)\citenamefont {Schuch},
  \citenamefont {Cirac},\ and\ \citenamefont {Verstraete}}]{Schuch08}%
  \BibitemOpen
  \bibfield  {author} {\bibinfo {author} {\bibfnamefont {N.}~\bibnamefont
  {Schuch}}, \bibinfo {author} {\bibfnamefont {J.~I.}\ \bibnamefont {Cirac}}, \
  and\ \bibinfo {author} {\bibfnamefont {F.}~\bibnamefont {Verstraete}},\
  }\href@noop {} {\bibfield  {journal} {\bibinfo  {journal} {Phys. Rev.
  Lett.},\ }\textbf {\bibinfo {volume} {100}},\ \bibinfo {pages} {250501}
  (\bibinfo {year} {2008})}\BibitemShut {NoStop}%
\end{thebibliography}

%merlin.mbs 2010-03-15 4.21a (PWD, AO, DPC)

\end{document}